\documentclass[11pt]{article}

\usepackage[margin=1in]{geometry}
\usepackage{lmodern}
\usepackage[utf8]{inputenc}
\usepackage[T1]{fontenc}
\usepackage{amsmath,amssymb,amsthm}
\usepackage{graphicx}
\usepackage{hyperref}
\usepackage{authblk}

\newcommand{\orgname}[1]{#1}
\newcommand{\orgdiv}[1]{#1}
\newcommand{\orgaddress}[1]{#1}
\newcommand{\country}[1]{#1}
\newcommand{\state}[1]{#1}
\newcommand{\postcode}[1]{#1}

\newcommand{\corresp}[2][]{\noindent\textbf{#2}}
\newcommand{\journaltitle}[1]{}
\newcommand{\DOI}[1]{}
\newcommand{\pubyear}[1]{}
\newcommand{\copyrightyear}[1]{}
\newcommand{\access}[1]{}
\newcommand{\appnotes}[1]{}
\newcommand{\firstpage}[1]{}
\newcommand{\subtitle}[1]{}
\newcommand{\history}[1]{}
\newcommand{\received}[3][]{}
\newcommand{\revised}[3][]{}
\newcommand{\accepted}[3][]{}
\newcommand{\authormark}[1]{}
\newcommand{\present}[1]{}
\theoremstyle{plain}

\onecolumn 

\graphicspath{{Fig/}}

\theoremstyle{plain}%
\newtheorem{theorem}{Theorem}
\newtheorem{proposition}[theorem]{Proposition}%
\theoremstyle{definition}%
\newtheorem{remark}{Remark}%
\theoremstyle{remark}%

\usepackage{enumitem}

\begin{document}
	
	\title{Generalizing Condorcet's Jury Theorem to Social Networks}
	
	\author[1,3]{Dan Braha}
	\author[1,2]{Marcus A.M. de Aguiar}

	
	\affil[1]{\orgname{New England Complex Systems Institute}, \orgaddress{ \state{Cambridge, MA}, \country{United States of America}}}
	\affil[2]{\orgdiv{Instituto de F\'isica Gleb Wataghin}, \orgname{Universidade Estadual de Campinas},  \orgaddress{\postcode{13083-970}, \state{Campinas, SP}, \country{Brazil}}}
	\affil[3]{\orgdiv{University of Massachusetts}, \orgname{Darthmouth}, \orgaddress{\state{Massachusetts},, \country{United States of America}}}
	
	
	\date{}
\maketitle
\thispagestyle{empty}

	\abstract{We generalize Condorcet's jury theorem (CJT) to socially connected populations in which agents revise discrete choices on a network in the presence of zealots. Free agents receive privately informative signals about the correct alternative and, at each update, either retain their state or imitate a uniformly chosen neighbor (free or zealot). For finite networks, we derive closed-form stationary laws for vote counts, and we characterize the corresponding vote-share limits as the number of free voters tends to infinity. For majority rule---both in binary and multi-alternative settings---we obtain an exact accuracy limit in closed form via the regularized incomplete beta function. For plurality rule, we establish sharp closed-form lower bounds on accuracy, expressed in terms of regularized incomplete beta functions. Under an absolute-majority condition for the correct alternative, both majority and plurality accuracies strictly exceed the accuracy of any single voter, showing that informative signals, coupled through social interaction, are amplified at the group level. These results extend CJT beyond independence and provide closed-form accuracy benchmarks for networked decision systems in social, biological, and engineered settings.}


\corresp[$\ast$]{Corresponding author. \href{mailto:braha@necsi.edu}{braha@necsi.edu}}

	\section{Introduction}

Condorcet's jury theorem (CJT) offers a foundational statistical account of collective decision-making: if each voter independently has a better-than-random chance of being correct, then a simple majority is more likely to select the correct option than any individual, with this advantage growing with group size and approaching certainty as the group becomes large. First articulated by Condorcet in 1785 \cite{r1}, the result was revived by Black \cite{r2}, independently rederived by Kazmann \cite{r3}, and named the ‘Condorcet jury theorem' by Grofman \cite{r4}. It is now a cornerstone of social choice theory \cite{r5,r6,r7,r8, r9, r10, r11, r12, r13, r14, r15, r16, r17, r18, r19, r20}, informing debates in philosophy (knowledge aggregation), political science (justifications for democracy and broad franchise), law (jury size and composition), and organizational decision-making across economics, management, marketing, and finance \cite{r5}.

Classical CJT relies on stringent assumptions: a binary choice with an objectively correct answer, common interests, homogeneous voter competence $p > 0.5$, independent private signals, and sincere (non-strategic) voting under simple majority. Relaxations of these assumptions reveal where information aggregation can fail. Under independence with heterogeneous competences, the non-asymptotic claim need not hold uniformly; nevertheless, majority accuracy typically exceeds the average voter competence and can, in some cases, exceed that of the best member, while the asymptotic claim holds under mild conditions on the competence distribution \cite{r6, r9, r11, r12, r20}. When individual competences are known and signals are independent, the optimal aggregation rule is a weighted majority with log-odds weights (with an appropriate threshold), which generalizes simple majority \cite{r16}. Under supermajority rules, when average voter competence exceeds the required vote share, sufficiently large electorates are more likely to choose the correct alternative than a randomly selected voter \cite{r17}. Departures from independence introduce correlation-sensitive phenomena: if voters share information sources or follow opinion leaders \cite{r11, r13, r18, r19}, positive correlation between voter decisions can erode or even reverse the Condorcet effect, whereas certain forms of negative correlation can strengthen it. Even so, conditions have been identified under which “wisdom of crowds” persists despite dependence, for example through explicit upper bounds on positive correlation under which majority decisions still outperform individual decisions, with accuracy that increases in larger groups.

A complementary literature studies social influence before voting. In sequential settings, informational cascades can generate herding on incorrect actions despite informative private signals \cite{r21}. In networked communication with naïve (DeGroot) updating, convergence to truth in large societies requires that no agent retains asymptotically dominant influence, a structural property of the graph \cite{r22}. Related Bayesian-learning models over networks identify topologies that permit—or preclude—efficient information aggregation \cite{r23}. Finally, in majority-dynamics models where agents iteratively adopt local majorities prior to a final vote, network structure can either block aggregation or guarantee unanimity when initial biases are favorable (e.g., on expander graphs), underscoring that how information flows is as pivotal as how it is counted \cite{r24}.

Taken together, these extensions, spanning statistical, epistemic, and social-choice perspectives, highlight both the power and limits of the classical CJT: large, independent, well-informed groups tend to excel at truth-tracking, but strategic behavior, correlation, and networked information can constrain accuracy. Despite this progress, models that make social influence and dependence explicit (e.g., via network-mediated interactions) and that yield closed-form characterizations of collective accuracy, allowing direct benchmarking against individual accuracy under realistic departures from independence, remain scarce.

\section{Motivation}\label{sec2}
	
	We develop a general framework of social influence that bridges a complex-systems perspective on collective human dynamics with formal models of collective decision-making. Within this framework, we extend the Condorcet Jury Theorem (CJT) to socially connected networks, showing that---even under peer influence---networked
	decision-making can achieve higher collective accuracy than individual voting, thereby moving beyond the assumption of independence. Our analysis builds on the classical voter model \cite{r25,r26,r27,r28,r29}, and in particular on voter models
	with fixed agents (“zealots”) \cite{r30,r31,r32,r33,r34,r35,r36,r37}. We make two contributions. First, for the binary (two-zealot) model introduced in our prior work \cite{r35,r38,r39, r40}, we give a new, direct derivation of the stationary distribution via a birth--death/detailed-balance argument and use it to obtain a generalized CJT. Second, we introduce a multistate (\(m\)-zealot) model, derive its closed-form stationary distribution, and prove Condorcet--type lower bounds for both majority and plurality of the correct alternative. For clarity, we begin with the binary case as a motivating baseline; the full multistate theory is presented in the Results section
	and recovers the binary results when \(m=2\).
	
	We consider an undirected network (graph) \(G=(\mathcal{V},E)\), where nodes \(\mathcal{V}\) are voters and edges \(E\) encode
	pairwise interactions. Each of the \(n+\alpha+\beta\) nodes holds a binary state: \(1\) (correct) or \(0\) (incorrect). There are \(\alpha\) nodes fixed in state \(1\) and \(\beta\) nodes fixed in state \(0\); we refer to these fixed-state nodes as
	\emph{zealots}. The remaining \(n\) nodes are \emph{free voters} whose states evolve via peer influence along the edges.
	
	The network evolves in discrete time. At each time step, a free node is selected uniformly at random, and its state is updated as follows: with probability $\ell$ it remains unchanged (\emph{retention/inertia}); with probability $1-\ell$ it copies the state of a uniformly chosen neighbor (\emph{imitation}). In the long-time limit $t \rightarrow \infty$, the alternative supported by more than half of the free voters is selected as the consensus outcome.
	
	For a fully connected network (complete graph), the free voters are exchangeable, so the system's state is fully determined by the number \(k\in\{0,\dots,n\}\) of free voters in state \(1\). The same reduction holds for a fully mixed (random-matching) population: under uniform pairings, the dynamics depend only on \(k\), and the one-step transition probabilities coincide with those for the complete graph. Denote by \(\sigma_k\) the macrostate with \(k\) free voters in state \(1\) (and \(n-k\) in state \(0\)). Let \(X_t\in\{0,\dots,n\}\) denote the number of free voters in state \(1\) at time \(t\), and set \(P_t(k):=P\{X_t=k\}\). Because a single free voter updates at each step, \(\{X_t\}\) is a birth–death chain: \(P_{t+1}(k)\) depends only on \(P_t(k-1)\), \(P_t(k)\), and \(P_t(k+1)\). The evolution is 
	\begin{equation*} 
		\begin{array}{ll} P_{t+1}(k) &= \displaystyle{P_t(k)\left\{\ell+		\frac{1-\ell} {n(n+\beta+\alpha-1)} \left[ k(k+\alpha-1) + (n-k)(n+	\beta-k-1) \right] \right\} +} \\ \\ & P_t(k-1) 	\displaystyle{\frac{1-\ell}{n(n+\beta+\alpha-1)} (k+\alpha-1)(n-k+1)} \, + \\ \\ & P_t(k+1) \displaystyle{\frac{1-\ell}{n(n+\beta+\alpha-1)} (k+1)(n+\beta-k-1)} \;, \end{array} 
	\end{equation*}

	This expression can be read term by term. For example, the third line is the probability that the system is in \(\sigma_{k+1}\) at time \(t\), multiplied by the probability that a voter in state \(1\) is selected for update, \((k+1)/n\), by the probability that the selected update is \emph{imitative} (rather than retention), \(1-\ell\), and by the conditional probability that it copies a voter in state \(0\), \(\bigl(n-k-1+\beta\bigr)/\bigl(n-1+\alpha+\beta\bigr)\). The other terms are analogous. (At the boundaries \(k=0\) and \(k=n\), the out-of-range terms are understood to be absent.)
	
	At stationarity, let \(X\) denote the number of free voters choosing the correct alternative (state \(1\)). In prior work \cite{r35,r38,r39,r40}, it was shown via hypergeometric generating functions that the stationary probability of observing \(X=k\), independent of the initial configuration, is:
	\begin{equation*}
		P\left(X=k\right)=\frac{\left(\genfrac{}{}{0pt}{0}{\alpha +k-1}k\right)\left(\genfrac{}{}{0pt}{0}{n+\beta
				-k-1}{n-k}\right)}{\left(\genfrac{}{}{0pt}{0}{n+\alpha +\beta -1}n\right)}
		\label{eq1}
	\end{equation*}
	where the influence parameters $\alpha$ and $\beta$ are generalizable to positive real values through the gamma function (generalized binomial coefficients). In particular, when $\alpha=\beta=1$, $P(X=k)=1/(n+1)$ for $k=0,1,\dots,n$ (uniform), a feature specific to fully connected networks. For completeness, we later give a direct birth--death (detailed-balance) derivation; see Proposition~2.
	
	As a baseline, suppose the free voters act independently and sample only from the
	\(\alpha+\beta\) fixed voters. Then each free voter selects the correct alternative
	(state \(1\)) with probability \(p=\alpha/(\alpha+\beta)\), the \emph{individual accuracy}.
	Consequently, the number of free voters in state \(1\) is \(\mathrm{Binomial}(n,p)\), with
	
	\begin{equation*}\label{eq:baseline-binom}
		P(X=k)=\binom{n}{k}\,p^{k}(1-p)^{\,n-k}, \qquad
		p=\frac{\alpha}{\alpha+\beta}.
	\end{equation*}

	According to the classical form of the CJT \cite{r6, r11}, independent voters with
	\(p>1/2\), under majority vote, yield
	\begin{equation*}\label{eq:cjt-majority}
		P_n(\text{correct decision}) = P\!\left(X \ge \left\lfloor \tfrac{n}{2} \right\rfloor + 1\right)
		= \sum_{k=\lfloor n/2 \rfloor + 1}^{n} \binom{n}{k}\, p^{k} (1-p)^{\,n-k} \;>\; p,
	\end{equation*}
	with \(P_n(\text{correct decision})\to 1\) as \(n\to\infty\). CJT thus demonstrates that a
	group's collective accuracy surpasses individual accuracy provided that individual
	signals are better than random.

	In contrast, when peer influence is incorporated via a fully connected network, the analogous inequality
	\begin{equation*}\label{eq:peer-majority}
		P_n(\text{correct decision})
		= \sum_{k=\lfloor n/2 \rfloor + 1}^{n}
		\frac{\displaystyle\binom{\alpha+k-1}{k}\,\binom{n+\beta-k-1}{\,n-k\,}}
		{\displaystyle\binom{n+\alpha+\beta-1}{\,n\,}}
		\;>\; p
	\end{equation*}
	does not hold uniformly for all \(n\), though numerical evidence suggests it holds asymptotically for large \(n\). We provide a rigorous analysis of the \(n\to\infty\) regime for the binary case and, crucially, extend the theory to \(m\) alternatives, where we derive a closed-form stationary distribution and prove majority/plurality performance guarantees (see Results). These findings show that peer influence can amplify collective accuracy beyond that of any single above-chance individual—even with multiple alternatives (under conditions made precise in the Results section).

	\section{Results}\label{sec3}
	
	\subsection{Binary zealots ($m=2$)}
	\label{subsec:binary}

We establish a generalized Condorcet Jury Theorem (CJT) for fully mixed networks with zealots, and then derive auxiliary propositions used in its proof.\par\medskip

\begin{theorem}
	\label{theorem1}
	\noindent \textbf{(Generalized CJT on fully mixed networks with binary zealots)}.
	Let $n\in\mathbb{N}$, $\alpha,\beta>0$, and $\ell\in[0,1)$. Consider a complete graph with $n+\alpha+\beta$ nodes, where $n$ are free voters, $\alpha$ are zealots fixed at state $1$ (``correct''), and $\beta$ are zealots fixed at state $0$ (``incorrect''). At each time step a free node is selected uniformly; with probability $\ell$ (\emph{retention/inertia}) it keeps its current state, and with probability $1-\ell$ (\emph{imitation}) it copies the state of a uniformly chosen neighbor (free or zealot). Let $X_t\in\{0,1,\dots,n\}$ be the number of free voters in state $1$ at time $t$.
	\begin{enumerate}[label=(\alph*), leftmargin=2em]
	\item The chain $\{X_t\}_{t\ge0}$ on $\{0,1,\dots,n\}$ is irreducible and aperiodic, with a unique stationary distribution (independent of $\ell$). In stationarity, $X\sim \mathrm{Beta\text{-}Binomial}(n;\alpha,\beta)$.
	
	\item Let $X$ have the stationary distribution of $\{X_t\}$, and assume $\alpha>\beta>0$. The probability of selecting the correct alternative by simple majority,
	\[
	P_n(\text{correct decision}) \;=\; P\!\left( X \ge \left\lfloor \frac{n}{2}\right\rfloor + 1 \right)
	\]
	satisfies, as $n\to\infty$, 
	\[
	\lim_{n\to\infty}P_n\!\left(\text{correct decision}\right)\;>\;\frac{\alpha}{\alpha+\beta},
	\]
	where $\alpha/(\alpha+\beta)$ is the (privately biased) probability that a free voter selects the correct alternative (i.e., chooses $1$).
	\end{enumerate}
\end{theorem}

The next proposition establishes part (a) of Theorem 1.\par\medskip

\begin{proposition}
	\label{proposition1}
	(Stationary distribution for the binary zealot voter model on fully mixed networks).
	Let $n\in\mathbb{N}$, $\alpha,\beta>0$, and $\ell\in[0,1)$. Consider the Markov chain $\{X_t\}_{t\ge0}$ described in Theorem 1. Then $\{X_t\}$ is an irreducible, aperiodic birth--death chain whose unique stationary distribution is
	\[
	\pi(k)=\binom{n}{k}\,\frac{(\alpha)_k\,(\beta)_{n-k}}{(\alpha+\beta)_n}
	=\frac{\displaystyle\binom{\alpha+k-1}{k}\,\binom{n+\beta-k-1}{\,n-k\,}}
	{\displaystyle\binom{n+\alpha+\beta-1}{\,n\,}},\qquad k=0,1,\dots,n,
	\]
	where $(a)_m:=a(a+1)\cdots(a+m-1)$ is the Pochhammer (rising) factorial with $(a)_0:=1$. The generalized binomial form extends to all real $\alpha,\beta>0$ via the gamma function. For $\alpha=0$ (resp.\ $\beta=0$), the unique stationary distribution is the point mass at $k=0$ (resp.\ at $k=n$).
\end{proposition}

\begin{proof}
	Because the interaction graph is complete, free voters are exchangeable and $\{X_t\}$ is a birth--death chain. From a state with $k$ ones among the $n$ free voters, the one-step probabilities are
	\[
	P_{k,k+1}=(1-\ell)\,\frac{n-k}{n}\cdot\frac{k+\alpha}{\,n-1+\alpha+\beta\,},\qquad
	P_{k,k-1}=(1-\ell)\,\frac{k}{n}\cdot\frac{n-k+\beta}{\,n-1+\alpha+\beta\,},
	\]
	and $P_{k,k}=1-P_{k,k+1}-P_{k,k-1}$, with the obvious boundary conventions. For $\alpha,\beta>0$ and $\ell<1$, we have $P_{k,k+1}>0$ for $k<n$ and $P_{k,k-1}>0$ for $k>0$, hence irreducible and aperiodic.
	
	Write the transition kernel as
	\[
	P=\ell I+(1-\ell)Q,
	\]
	where $Q$ is the \emph{copy-step} kernel (set $\ell=0$ in the above formulas), i.e.
	\[
	Q_{k,k+1}=\frac{n-k}{n}\cdot\frac{k+\alpha}{\,n-1+\alpha+\beta\,},\qquad
	Q_{k,k-1}=\frac{k}{n}\cdot\frac{n-k+\beta}{\,n-1+\alpha+\beta\,},\qquad
	Q_{k,k}=1-Q_{k,k+1}-Q_{k,k-1}.
	\]
	Since $P=\ell I+(1-\ell)Q$, if $\pi Q=\pi$ then $\pi P=\ell\pi+(1-\ell)\pi=\pi$. Hence the stationary distribution does not depend on $\ell$, and it suffices to determine the stationary law for $Q$.
	
	For a birth--death chain, stationarity is equivalent to detailed balance:
	\[
	\pi(k)\,Q_{k,k+1}=\pi(k+1)\,Q_{k+1,k}\qquad (k=0,1,\dots,n-1).
	\]
	Therefore
	\[
	\frac{\pi(k+1)}{\pi(k)}=\frac{Q_{k,k+1}}{Q_{k+1,k}}
	=\frac{(n-k)(k+\alpha)}{(k+1)(n-k-1+\beta)}.
	\]
	Solving this first-order recurrence yields
	\[
	\pi(k)=\pi(0)\,\prod_{i=0}^{k-1}\frac{(n-i)(i+\alpha)}{(i+1)(n-i-1+\beta)}
	=\pi(0)\,\binom{n}{k}\,(\alpha)_k\,\frac{(\beta)_{n-k}}{(\beta)_n}.
	\]
	To identify $\pi(0)$, use Vandermonde's identity in Pochhammer form,
	\[
	\sum_{k=0}^{n}\binom{n}{k}\,(\alpha)_k\,(\beta)_{n-k}=(\alpha+\beta)_n,
	\]
	which follows by multiplying the binomial series $(1-z)^{-\alpha}=\sum_{k\ge0}(\alpha)_k z^k/k!$ and $(1-z)^{-\beta}=\sum_{m\ge0}(\beta)_m z^m/m!$ and equating the $z^n$ coefficients after multiplying by $n!$. Hence
	\[
	\pi(k)=\binom{n}{k}\,\frac{(\alpha)_k\,(\beta)_{n-k}}{(\alpha+\beta)_n}.
	\]
	Finally, using $(a)_m=\Gamma(a+m)/\Gamma(a)=m!\,\binom{a+m-1}{m}$ gives the equivalent generalized binomial form displayed in the statement. Uniqueness of the stationary distribution and convergence to it follow from the chain's irreducibility and aperiodicity; the degenerate cases $\alpha=0$ or $\beta=0$ are immediate from the transition rules.
\end{proof}

Part (b) of Theorem \ref{theorem1} follows from the next two propositions. To analyze the asymptotic regime $n\to\infty$, let $V=X/n$ denote the stationary fraction of free voters choosing the correct alternative. We then establish the following result:\par\medskip

\begin{proposition}
	\label{proposition3}
	(Beta limit for the stationary vote share).
	Fix $\alpha,\beta>0$. On the complete graph with $n$ free voters, let $X$ denote
	(at stationarity) the number of free voters who select the correct alternative, and
	set $V_n:=X/n$. Then
	\[
	V_n \overset{p}{\longrightarrow} V, \qquad V \sim \mathrm{Beta}(\alpha,\beta).
	\]
	
	Consequently,
	\[
	\lim_{n\to\infty} P_n(\text{correct decision})
	= \lim_{n\to\infty} P\!\big(V_n\ge \tfrac12\big)
	= P\!\big(V\ge \tfrac12\big)
	= 1 - I_{1/2}(\alpha,\beta),
	\]
	where $I_x(\alpha,\beta)$ is the \emph{regularized incomplete Beta function} \cite{r41},
	\[
	I_x(\alpha,\beta)
	:= \frac{1}{B(\alpha,\beta)}\int_0^x t^{\alpha-1}(1-t)^{\beta-1}\,dt,
	\qquad
	B(\alpha,\beta)=\frac{\Gamma(\alpha)\Gamma(\beta)}{\Gamma(\alpha+\beta)}.
	\]
\end{proposition}

\begin{proof}
	By Theorem~1(a), at stationarity $X$ has the Beta--Binomial$(n;\alpha,\beta)$ law.
	Equivalently, $X$ admits the mixture representation \cite{r41}:
	\[
	V \sim \mathrm{Beta}(\alpha,\beta), \qquad X \mid V \sim \mathrm{Binomial}(n,V),
	\]
	with $V$ independent of the binomial sampling. Define $V_n:=X/n$.
	
	Given $V$, $X$ is the sum of $n$ i.i.d.\ $\mathrm{Bernoulli}(V)$ trials, so by the (conditional) law of large numbers,
	\[
	V_n=\frac{X}{n} \;\xrightarrow{\ \text{a.s.}\ }\; V \quad\text{(given $V$)}.
	\]
	Therefore, unconditionally,
	\[
	V_n \overset{p}{\longrightarrow} V .
	\]
	
	Since $V\sim\mathrm{Beta}(\alpha,\beta)$, this convergence in probability implies
	\[
	V_n \overset{d}{\longrightarrow} V \ \sim\ \mathrm{Beta}(\alpha,\beta).
	\]
	Because the Beta cdf is continuous at $1/2$,
	\begin{equation*}
\lim_{n\to\infty} P\!\big(V_n\ge \tfrac12\big)
= P\!\big(V\ge \tfrac12\big)
= 1 - I_{1/2}(\alpha,\beta).
\end{equation*}

\end{proof}
	
Recall that \(\mu:=\alpha/(\alpha+\beta)\) is the individual (private-signal) accuracy—the
probability that a free voter selects the correct alternative (state \(1\)). For general
\(\alpha,\beta\) with \(\mu>1/2\), the limiting majority accuracy is
\[
\lim_{n\to\infty} P_n(\text{correct decision})
= P\!\left(V\ge \tfrac12\right)
= 1 - I_{1/2}(\alpha,\beta)\in(1/2,1),
\]
so it need not equal \(1\) (in contrast with the classical Condorcet Jury Theorem under
independent votes). Nevertheless, we show below that it strictly exceeds the individual
accuracy:
\[
\lim_{n\to\infty} P_n(\text{correct decision})
= P\!\left(V\ge \tfrac12\right)\;>\;\frac{\alpha}{\alpha+\beta}.
\]
This bound is established analytically below, proving part (b) of Theorem~1.\par\medskip

\begin{proposition}
	\label{proposition4}
(Collective accuracy exceeds individual accuracy).
Let $\alpha, \beta > 0$ with $\alpha > \beta$. Then
\[
P(V \geq \tfrac{1}{2}) \;=\; 1 - I_{1/2}(\alpha, \beta) \;>\; \frac{\alpha}{\alpha + \beta},
\]
which is equivalent to
\[
P(V \leq \tfrac{1}{2}) \;=\; I_{1/2}(\alpha, \beta) \;<\; \frac{\beta}{\alpha + \beta}.
\]	
\end{proposition}

\begin{proof}
	Let $V\sim\mathrm{Beta}(\alpha,\beta)$ with density
	\[
	f_{\alpha,\beta}(u)\;=\;\frac{u^{\alpha-1}(1-u)^{\beta-1}}{B(\alpha,\beta)},\qquad u\in(0,1).
	\]
	We use the symmetry $f_{\alpha,\beta}(1-u)=f_{\beta,\alpha}(u)$ and $B(\alpha,\beta)=B(\beta,\alpha)$.
	
	Since the beta law is continuous, $P(V\ge 1/2)=P(V>1/2)$. Consider
	\[
	P(V>1/2)-\mathbb{E}[V]
	=\int_{1/2}^{1} f_{\alpha,\beta}(u)\,du-\int_0^1 u\,f_{\alpha,\beta}(u)\,du.
	\]
	Split the expectation at $1/2$ and regroup:
	\[
	P(V>1/2)-\mathbb{E}[V]
	= \int_{1/2}^{1} (1-u)\,f_{\alpha,\beta}(u)\,du
	- \int_{0}^{1/2} u\,f_{\alpha,\beta}(u)\,du.
	\]
	In the first integral substitute $v=1-u$ (so $u\in[1/2,1]\iff v\in[1/2,0]$, $du=-dv$):
	\[
	\int_{1/2}^{1} (1-u)\,f_{\alpha,\beta}(u)\,du
	= \int_{0}^{1/2} v\, f_{\alpha,\beta}(1-v)\,dv
	= \int_{0}^{1/2} v\, f_{\beta,\alpha}(v)\,dv.
	\]
	Therefore
	\[
	P(V>1/2)-\mathbb{E}[V]
	= \int_{0}^{1/2} v\,[\,f_{\beta,\alpha}(v)-f_{\alpha,\beta}(v)\,]\,dv.
	\]
	Compute the pointwise ratio (the beta constants cancel):
	\[
	\frac{f_{\beta,\alpha}(v)}{f_{\alpha,\beta}(v)}
	=\frac{v^{\beta-1}(1-v)^{\alpha-1}}{v^{\alpha-1}(1-v)^{\beta-1}}
	=\Big(\frac{1-v}{v}\Big)^{\alpha-\beta}.
	\]
	For $v\in(0,1/2)$ we have $(1-v)/v>1$, and since $\alpha>\beta$, it follows that
	$f_{\beta,\alpha}(v)-f_{\alpha,\beta}(v)>0$ on $(0,1/2)$. Because $v>0$ there as well, the integrand is strictly positive, hence
	\[
	P(V>1/2)-\mathbb{E}[V]\;>\;0.
	\]
	Finally, $\mathbb{E}[V]=\alpha/(\alpha+\beta)$, so
	\[
	P(V>1/2)\;>\;\frac{\alpha}{\alpha+\beta}.
	\]
	Equivalently, using $P(V\le 1/2)=I_{1/2}(\alpha,\beta)$,
	\[
	P(V\ge \tfrac12)=1-I_{1/2}(\alpha,\beta)\;>\;\frac{\alpha}{\alpha+\beta},
	\qquad
	P(V\le \tfrac12)=I_{1/2}(\alpha,\beta)\;<\;\frac{\beta}{\alpha+\beta}.
	\] 
\end{proof}


\begin{remark}
	\label{remark1}
	Proposition \ref{proposition4} shows that when individual signals are even slightly biased toward
	the correct alternative (\(\mu>1/2\)), the majority's probability of choosing
	correctly \emph{strictly exceeds} that of any single voter. In other words,
	weakly informative signals, coupled through social interaction, are amplified at
	the group level.
\end{remark}

\begin{remark}
	\label{remark2}
	It is instructive to observe that the classical Condorcet Jury Theorem (CJT) is
	recovered from Proposition \ref{proposition3} in a suitable limit. Consider the \emph{many fixed
		voters} regime in which \(\alpha,\beta\gg 1\) while the ratio
	\[
	\mu:=\frac{\alpha}{\alpha+\beta}\in(0,1)
	\]
	is held fixed. In this regime, the influence of the external (fixed) voters
	dominates the dynamics of the free voters \cite{r35}; dependencies among
	free-voter decisions become negligible, so the free voters behave as if
	independent. For the stationary vote share \(V\sim\mathrm{Beta}(\alpha,\beta)\),
	\[
	\operatorname{Var}(V)
	=\frac{\alpha\beta}{(\alpha+\beta)^2(\alpha+\beta+1)}
	\longrightarrow 0 \qquad (\alpha+\beta\to\infty).
	\]
	By Chebyshev's inequality, for any \(\varepsilon>0\),
	\[
	P\!\big(|V-\mu|>\varepsilon\big)
	\le \frac{\operatorname{Var}(V)}{\varepsilon^2}
	\longrightarrow 0 ,
	\]
	so \(V \xrightarrow{\ p\ } \mu\) (i.e., \(V\) concentrates at \(\mu\)). Consequently,
	\[
	P\!\left(V\ge \tfrac12\right) \longrightarrow 1
	\quad \text{whenever } \mu>\tfrac12,
	\]
	thereby recovering the classical CJT in this limit.
	
\end{remark}

\subsection{Multistate zealots ($m \geq 2$)}
\label{subsec:multistate}
\begin{theorem}
	\label{theorem7}
	(Generalized CJT on fully mixed networks with $m$ zealot types).
	Let $m\ge 2$, $n\in\mathbb{N}$, and $\alpha_i>0$ for $i=1,\dots,m$. Consider a complete graph with $n+\sum_{i=1}^m \alpha_i$ nodes, where $n$ are free voters and, for each $i$, there are $\alpha_i$ zealots fixed in state $i$. At each time step, a free node is selected uniformly at random; with probability $\ell\in[0,1)$ (\textit{retention/inertia}) it keeps its current state, and with probability $1-\ell$ (\textit{imitation}) it copies the state of a uniformly chosen neighbor (free or zealot). \emph{State $1$ is the correct alternative.} Let
	\[
	X_t=(X_{t,1},\dots,X_{t,m})\in \mathcal{S}_n:=\Big\{k\in\mathbb{Z}_{\ge0}^m:\ \sum_{i=1}^m k_i=n\Big\}
	\]
	denote the vector of free-voter counts in each state at time $t$, and write $\alpha_0:=\sum_{i=1}^m \alpha_i$.
	
	\begin{enumerate}[label=(\alph*), leftmargin=2em]
		
		\item The chain $\{X_t\}_{t\ge0}$ is irreducible and aperiodic on $\mathcal{S}_n$, with a unique stationary distribution (independent of $\ell$). In stationarity, \(X \sim \mathrm{Dirichlet\text{-}Multinomial}(n;\alpha_1,\dots,\alpha_m)\);
		for each \(j\), \(X_j \sim \mathrm{Beta\mbox{-}Binomial}(n;\alpha_j,\alpha_0-\alpha_j)\).

		\item Let $\mathrm{Maj}_n=\{X_1\ge \lfloor n/2\rfloor+1\}$ denote the event that a strict majority of free voters choose the correct state. Then, in exact finite-$n$ terms,
		\[
		\begin{aligned}
			P(\mathrm{Maj}_n)
			&= \sum_{k=\lfloor n/2\rfloor+1}^{n}
			\binom{n}{k}\,
			\frac{(\alpha_1)_k\,(\alpha_0-\alpha_1)_{\,n-k}}{(\alpha_0)_{\,n}} \\
			&= \sum_{k=\lfloor n/2\rfloor+1}^{n}
			\binom{n}{k}\;
			\frac{\displaystyle \binom{\alpha_1+k-1}{\,k\,}\;
				\binom{n+\alpha_0-\alpha_1-k-1}{\,n-k\,}}
			{\displaystyle \binom{\alpha_0+n-1}{\,n\,}}\,.
		\end{aligned}
		\]
		
		\noindent Moreover,
		\[
		\lim_{n\to\infty} P(\mathrm{Maj}_n)
		= 1 - I_{1/2}\!\big(\alpha_1,\alpha_0-\alpha_1\big).
		\]
		If \(\alpha_1 > \alpha_0 - \alpha_1\), then \(\lim_{n\to\infty} P(\mathrm{Maj}_n) > \alpha_1/\alpha_0\), where \(\alpha_1/\alpha_0\) is the probability that a free, independent voter selects the correct alternative \(1\) (see Remark~3).

		\item Let \(\mathrm{Plu}_n\) denote the event that state \(1\) wins by plurality, that is,
		\(\mathrm{Plu}_n=\{\,X_1>\max_{j\neq 1} X_j\,\}\). Then
		\[
		\lim_{n\to\infty} P(\mathrm{Plu}_n)
		\;\ge\;
		\max\!\left\{\,1-I_{1/2}\!\big(\alpha_1,\alpha_0-\alpha_1\big),\ 
		1-\sum_{j=2}^{m} I_{1/2}\!\big(\alpha_1,\alpha_j\big)\,\right\}.
		\]
		If \(\alpha_1 > \alpha_0 - \alpha_1\), then
		\(\lim\limits_{n\to\infty} P(\mathrm{Plu}_n) \ge 1 - I_{1/2}(\alpha_1,\alpha_0-\alpha_1) > \alpha_1/\alpha_0\), with \(\alpha_1/\alpha_0\) interpreted as in part~(b).
	\end{enumerate}
	
\end{theorem}

The next proposition proves part (a) of Theorem \ref{theorem7}; the proofs of parts (b) and (c) follow later.\par\medskip

\begin{proposition}
	\label{proposition8}
	(Stationary distribution for the multistate zealot voter model on fully mixed networks).
	Let $m\ge2$, $n\in\mathbb{N}$, and $\alpha_i>0$ for $i=1,\dots,m$. 
	Consider the Markov chain $\{X_t\}_{t\ge0}$ described in Theorem \ref{theorem7}. 
	Then $\{X_t\}_{t\ge0}$ is an irreducible, aperiodic Markov chain on $\mathcal{S}_n$ whose unique stationary distribution, independent of $\ell$, is
	\noindent
	\begin{equation}
		\pi(k_1,\dots,k_m)\;=\;
		\frac{\displaystyle\prod_{i=1}^m \binom{\alpha_i+k_i-1}{\,k_i\,}}
		{\displaystyle \binom{\alpha_0+n-1}{\,n\,}},
		\qquad (k_1,\dots,k_m)\in\mathcal{S}_n,\quad \alpha_0 := \sum_{i=1}^m \alpha_i,
		\label{DR_voter}
	\end{equation}
	where generalized binomial coefficients are understood in the Gamma-function sense. Equivalently,
	\[
	\begin{aligned}
		\pi(k_1,\dots,k_m)
		&= \frac{n!}{k_1!\cdots k_m!}\,
		\frac{(\alpha_1)_{k_1}\cdots(\alpha_m)_{k_m}}{(\alpha_1+\cdots+\alpha_m)_{n}} \\
		&= \frac{n!}{k_1!\cdots k_m!}\,
		\frac{\Gamma(\alpha_0)}{\Gamma(\alpha_0+n)}
		\prod_{i=1}^m \frac{\Gamma(\alpha_i+k_i)}{\Gamma(\alpha_i)}.
	\end{aligned}
	\]
	That is, \(X\sim\mathrm{Dirichlet\!-\!Multinomial}
	\big(n;\allowbreak \alpha_1,\dots,\allowbreak \alpha_m\big)\);
	for each \(j\), the marginal \(X_j\sim \mathrm{Beta\!-\!Binomial}
	\big(n;\allowbreak \alpha_j,\allowbreak \alpha_0-\alpha_j\big)\);
	moreover, when \(\alpha_1=\cdots=\alpha_m=1\), the distribution is
	uniform on \(\mathcal{S}_n\).
\end{proposition}

\begin{proof}
	Write $P=\ell I+(1-\ell)Q$, where $Q$ is the pure imitation kernel ($\ell=0$). From a state $k=(k_1,\dots,k_m)$, an elementary move $k\to k+e_i-e_j$ ($i\neq j$) occurs under $Q$ with probability
	\[
	Q\big(k\to k+e_i-e_j\big)=\frac{k_j}{n}\cdot \frac{k_i+\alpha_i}{\,n-1+\alpha_0\,},
	\]
	since we pick a free $j$-node with probability $k_j/n$ and it copies an $i$-neighbor (free or zealot) with probability $(k_i+\alpha_i)/(n-1+\alpha_0)$.
	We \emph{propose} the stationary weights in the generalized-binomial form
	\[
	w(k)\;:=\;\prod_{i=1}^m \binom{\alpha_i+k_i-1}{\,k_i\,}.
	\]
	Consider $k' = k+e_i-e_j$ with $i\neq j$. Using the elementary identities
	\[
	\frac{\binom{a+r}{\,r+1\,}}{\binom{a+r-1}{\,r\,}}=\frac{a+r}{r+1},
	\qquad
	\frac{\binom{a+(r-1)-1}{\,r-1\,}}{\binom{a+r-1}{\,r\,}}=\frac{r}{a+r-1},
	\]
	valid for all real $a>0$ and integers $r\ge1$, we obtain
	\[
	\frac{w(k')}{w(k)}
	=\frac{\binom{\alpha_i+k_i}{\,k_i+1\,}}{\binom{\alpha_i+k_i-1}{\,k_i\,}}
	\cdot
	\frac{\binom{\alpha_j+(k_j-1)-1}{\,k_j-1\,}}{\binom{\alpha_j+k_j-1}{\,k_j\,}}
	=\frac{\alpha_i+k_i}{k_i+1}\cdot \frac{k_j}{\alpha_j+k_j-1}.
	\]
	On the other hand, the reverse move $k'\to k$ has probability
	\[
	Q(k'\to k)=\frac{k_i+1}{n}\cdot \frac{k_j-1+\alpha_j}{\,n-1+\alpha_0\,},
	\]
	so the transition ratio satisfies
	\[
	\frac{Q(k\to k')}{Q(k'\to k)}
	=\frac{\frac{k_j}{n}\cdot \frac{k_i+\alpha_i}{n-1+\alpha_0}}{\frac{k_i+1}{n}\cdot \frac{k_j-1+\alpha_j}{n-1+\alpha_0}}
	=\frac{k_j}{k_i+1}\cdot \frac{\alpha_i+k_i}{\alpha_j+k_j-1}
	=\frac{w(k')}{w(k)}.
	\]
	Therefore $w(k)\,Q(k\to k')=w(k')\,Q(k'\to k)$ for all adjacent pairs, i.e.\ detailed balance holds for $Q$. Hence $\pi:=C\,w$ is stationary for $Q$, and thus for $P$ as well (since $\pi P=\ell\pi+(1-\ell)\pi=\pi$). Irreducibility and aperiodicity of $\{K_t\}$ follow from $\ell<1$ and $\alpha_i>0$ (all unit moves $e_i-e_j$ occur with positive probability).
	It remains to determine the normalizing constant $C$. Consider the generating function identity
	\[
	\sum_{r\ge0} \binom{\alpha_i+r-1}{\,r\,} z^r=(1-z)^{-\alpha_i}\qquad (|z|<1),
	\]
	multiply over $i=1,\dots,m$, and extract the coefficient of $z^n$:
	\[
	\sum_{k\in\mathcal{S}_n}\prod_{i=1}^m \binom{\alpha_i+k_i-1}{\,k_i\,}
	=\binom{\alpha_0+n-1}{\,n\,}.
	\]
	Thus $C=\binom{\alpha_0+n-1}{\,n\,}^{-1}$ and Eq.(\ref{DR_voter}) follows. Uniqueness of the stationary distribution and convergence to it are standard for finite irreducible and aperiodic Markov chains.
\end{proof}

The proof of Theorem \ref{theorem7}(b) appears below.\par\medskip
{%
\renewcommand{\proofname}{Proof of Theorem \ref{theorem7}(b) (Collective majority accuracy with $m$ zealot types)}
\begin{proof}
	\label{prooftheorem7b}
	Let $\mathrm{Maj}_n=\{X_1\ge \lfloor n/2\rfloor+1\}$ denote the event that a strict majority of free voters
	choose the correct state $1$. By part (a) of Theorem \ref{theorem7}, at stationarity
	\[
	X_1 \sim \mathrm{Beta\text{-}Binomial}\big(n;\alpha_1,\alpha_0-\alpha_1\big),
	\qquad \alpha_0:=\sum_{i=1}^m \alpha_i .
	\]
	By Theorem~1(a), the exact finite-$n$ majority probability is the
	Beta--Binomial upper tail:
	\[
	\begin{aligned}
		P(\mathrm{Maj}_n)
		&= \sum_{k=\lfloor n/2\rfloor+1}^{n}
		\binom{n}{k}\;
		\frac{\displaystyle \binom{\alpha_1+k-1}{\,k\,}\;
			\binom{n+\alpha_0-\alpha_1-k-1}{\,n-k\,}}
		{\displaystyle \binom{\alpha_0+n-1}{\,n\,}}\,.
	\end{aligned}
	\]
	If \(\alpha_1 > \alpha_0 - \alpha_1\) (equivalently, \(\alpha_1/\alpha_0 > \tfrac12\)), then, in the large-\(n\) limit,
	\[
	\lim_{n\to\infty} P(\mathrm{Maj}_n)
	= 1 - I_{1/2}\!\big(\alpha_1,\alpha_0-\alpha_1\big) > \frac{\alpha_1}{\alpha_0}.
	\]
	The equality follows from Proposition~3 (applied with \((\alpha,\beta)=(\alpha_1,\alpha_0-\alpha_1)\)), and the strict inequality from Proposition~4. As shown in Remark~3, \(\alpha_1/\alpha_0\) is the probability that a free, independent voter selects the correct alternative \(1\).
\end{proof}
} 
We will use the next proposition to prove part (c) of Theorem \ref{theorem7}.\par\medskip

\begin{proposition}
	\label{proposition9}
	(Dirichlet limit for the stationary vote shares).
	Let \(X=(X_1,\dots,X_m)\sim\mathrm{Dirichlet\text{-}Multinomial}\!\big(n;\alpha_1,\dots,\alpha_m\big)\) with \(\alpha_i>0\)
	and \(\alpha_0:=\sum_{i=1}^m \alpha_i\). Define the stationary proportions
	\[
	V_n:=\frac{X}{n}=(X_1/n,\dots,X_m/n)\in
	\Delta_m:=\big\{\,v\in[0,1]^m:\ \sum_{i=1}^m v_i=1\,\big\}.
	\]
	Then
	\[
	V_n \xrightarrow{\ p\ } \Theta
	\quad\text{and hence}\quad
	V_n \xrightarrow{\ d\ } \Theta,
	\]
	where \(\Theta\sim\mathrm{Dirichlet}(\alpha_1,\dots,\alpha_m)\), with density
	\[
	f(v)=\frac{\Gamma(\alpha_0)}{\displaystyle\prod_{i=1}^m \Gamma(\alpha_i)}
	\prod_{i=1}^m v_i^{\alpha_i-1},\qquad v\in\Delta_m.
	\]
\end{proposition}

\begin{proof}
Use the standard mixture representation of the Dirichlet--Multinomial:
\[
\Theta=(\Theta_1,\dots,\Theta_m)\sim\mathrm{Dirichlet}(\alpha_1,\dots,\alpha_m),
\qquad
X\mid \Theta \sim \mathrm{Multinomial}\big(n,\Theta\big),
\]
with $\Theta$ independent of the sampling.

Conditionally on $\Theta$, realize the $n$ free voters as i.i.d.\ categorical variables
$Z_1,\dots,Z_n\in\{1,\dots,m\}$ with $\Pr(Z_t=i\mid\Theta)=\Theta_i$. Then the counts are
\[
X_i=\sum_{t=1}^n I_{\{Z_t=i\}},\qquad i=1,\dots,m,
\]
so the proportions satisfy
\[
(V_n)_i=\frac{X_i}{n}=\frac{1}{n}\sum_{t=1}^n I_{\{Z_t=i\}}.
\]
Here $I_A$ denotes the indicator of an event $A$ (equal to $1$ if $A$ occurs and $0$ otherwise).

By the strong law of large numbers applied conditionally on $\Theta$,
\[
(V_n)_i \xrightarrow{\ \text{a.s.}\ } \Theta_i \quad \text{for each } i=1,\dots,m,
\]
hence $V_n \xrightarrow{\ \text{a.s.}\ } \Theta$ given $\Theta$. Taking expectations over $\Theta$ yields the unconditional convergence in probability:
\[
V_n \xrightarrow{\ p\ } \Theta.
\]
Since convergence in probability implies convergence in distribution, we also have
$V_n \overset{d}{\longrightarrow} \Theta$, and $\Theta$ has the Dirichlet$(\alpha_1,\dots,\alpha_m)$ density displayed above.
\end{proof}

We now prove Theorem \ref{theorem7} (c).
\par\medskip
{%
\renewcommand{\proofname}{Proof of Theorem~\ref{theorem7}(c) (Collective plurality accuracy with $m$ zealot types)}
\begin{proof}
	We first analyze the case that state 1 wins by majority. Recall that $\mathrm{Plu}_n=\{X_1>\max_{j\ne1}X_j\}$ and $\mathrm{Maj}_n=\{X_1\ge \lfloor n/2\rfloor+1\}$.
	
	\par\medskip
	\noindent\emph{Step 1.} Since $\mathrm{Maj}_n\subseteq \mathrm{Plu}_n$,
	\[
	P(\mathrm{Plu}_n)\ \ge\ P(\mathrm{Maj}_n)\quad\Longrightarrow\quad
	\lim_{n\to\infty}P(\mathrm{Plu}_n)\ \ge\ 1-I_{1/2}(\alpha_1,\alpha_0-\alpha_1),
	\]
	by part (b) of Theorem \ref{theorem7}.
	
	\par\medskip
	\noindent\emph{Step 2.}
	We next analyze the case that state 1 wins by plurality. By Proposition \ref{proposition9}, 
	\[
	V_n:=\frac{X}{n}\ \xrightarrow{\ p\ }\ \Theta,
	\]
	where \(\Theta\sim\mathrm{Dirichlet}(\alpha_1,\dots,\alpha_m)\). Hence for each fixed $j\ge2$,
	\[
	\frac{X_1}{n}\ \xrightarrow{\ p\ }\ \Theta_1,\qquad
	\frac{X_j}{n}\ \xrightarrow{\ p\ }\ \Theta_j.
	\]
	We use the following standard comparison property: if $A_n\to A$ and $B_n\to B$ in probability and
	$P(A=B)=0$, then $P(A_n<B_n)\to P(A<B)$. Applying this with $A_n=X_1/n$, $B_n=X_j/n$, $A=\Theta_1$, $B=\Theta_j$, and noting that a Dirichlet vector has a continuous
	density (so $P(\Theta_1=\Theta_j)=0$), we obtain
	\[
	\lim_{n\to\infty}P(X_1\le X_j)
	=\lim_{n\to\infty}P\!\big(\tfrac{X_1}{n}\le \tfrac{X_j}{n}\big)
	= P(\Theta_1\le \Theta_j).
	\]
	For the Dirichlet law,
	\[
	\frac{\Theta_1}{\Theta_1+\Theta_j}\sim \mathrm{Beta}(\alpha_1,\alpha_j)
	\quad\Rightarrow\quad
	P(\Theta_1\le \Theta_j)=P\!\Big(\frac{\Theta_1}{\Theta_1+\Theta_j}\le\tfrac{1}{2}\Big)=I_{1/2}(\alpha_1,\alpha_j).
	\]
	
	\par\medskip
	\noindent\emph{Step 3.}
	Note that the complement of plurality can be written as a union of pairwise comparison failures:
	\[
	\mathrm{Plu}_n^c=\{X_1\le \max_{j\ne1}X_j\}=\bigcup_{j=2}^m \{X_1\le X_j\}.
	\]
	Hence, by the union bound,
	\[
	P(\mathrm{Plu}_n)=1-P(\mathrm{Plu}_n^c)
	\ \ge\ 1-\sum_{j=2}^m P(X_1\le X_j).
	\]
	Taking limits and using Step~2,
	\[
	\lim_{n\to\infty} P(\mathrm{Plu}_n)
	\ \ge\ 1-\sum_{j=2}^{m}\lim_{n\to\infty}P(X_1\le X_j)
	\ =\ 1-\sum_{j=2}^{m} I_{1/2}(\alpha_1,\alpha_j).
	\]
	Together with Step~1, we conclude
	\begin{equation}\label{eq:plu-lb}
		\lim_{n\to\infty} P(\mathrm{Plu}_n)\ \ge\
		\max\Big\{\,1-I_{1/2}(\alpha_1,\alpha_0-\alpha_1)\,,\
		\ 1-\sum_{j=2}^{m} I_{1/2}(\alpha_1,\alpha_j)\,\Big\}.
	\end{equation}
	Finally, since the right-hand side of Eq.\eqref{eq:plu-lb} is at least
	\(1 - I_{1/2}(\alpha_1,\alpha_0-\alpha_1)\), Proposition~4 implies that it is
	strictly greater than \(\alpha_1/\alpha_0\) whenever \(\alpha_1 > \alpha_0 - \alpha_1\).
	Thus the single condition \(\alpha_1 > \alpha_0 - \alpha_1\) (equivalently,
	\(\alpha_1/\alpha_0 > \tfrac12\)) provides a simple sufficient criterion for the
	plurality lower bound in Theorem \ref{theorem7}(c) to strictly exceed the individual accuracy
	under independent voting, \(\alpha_1/\alpha_0\) (see Remark~3). \hfill\(\square\)
	
\end{proof}
}

We conclude with a brief remark establishing that, under independent voting, majority and plurality accuracy parallel the classical CJT as \(n\to\infty\).
\medskip
\begin{remark}
	\label{remark3}
	(Independent voters: majority and plurality accuracy).
	If free voters act independently and interact only with zealots, each free voter
	samples uniformly from the \(\alpha_0:=\sum_{i=1}^m \alpha_i\) zealots. Thus the free
	votes are i.i.d.\ with category probabilities
	\[
	q_i=\frac{\alpha_i}{\alpha_0},\qquad i=1,\dots,m,
	\]
	and the count vector is
	\[
	X=(X_1,\dots,X_m)\sim \mathrm{Multinomial}\!\big(n;\,q_1,\dots,q_m\big).
	\]
	
	\medskip
	\noindent\emph{Majority rule.}
	Under strict majority, only the dichotomy “\(1\) vs.\ not-\(1\)” matters:
	\[
	P(\text{majority for \(1\)})
	= P\!\left(X_1 \ge \left\lfloor \tfrac{n}{2}\right\rfloor+1\right)
	= \sum_{k=\lfloor n/2\rfloor+1}^{n} \binom{n}{k}\, q_1^{\,k}\,(1-q_1)^{\,n-k}.
	\]
	Consequently, if \(q_1>1/2\), then (as in the classical CJT)
	\[
	P(\text{majority for \(1\)})>q_1
	\quad\text{and}\quad
	P(\text{majority for \(1\)}) \xrightarrow[n\to\infty]{} 1.
	\]
	
	\medskip
	\noindent\emph{Plurality rule.}
	Under plurality, state \(1\) wins when it strictly exceeds every competitor:
	\[
	P(\text{plurality for \(1\)}) \;=\; P\!\big(X_1>\max_{j\ne 1}X_j\big),
	\qquad X\sim\mathrm{Multinomial}\big(n;\,q_1,\dots,q_m\big).
	\]
	While a closed form is generally unavailable, the law of large numbers yields
	\(\tfrac{X_i}{n}\xrightarrow{\ p\ }q_i\) for each \(i\). Assume, without loss of generality,
	that \(\alpha_1\neq \alpha_j\) for all \(j\neq 1\) (equivalently, \(q_1\neq q_j\)). Then
	\[
	P(\text{plurality for }1)\xrightarrow[n\to\infty]{}
	\begin{cases}
		1, & \text{if } q_1>\max_{j\ne 1} q_j,\\
		0, & \text{if } q_1<\max_{j\ne 1} q_j.
	\end{cases}
	\]
	
	(For completeness: if exactly \(r\) categories tie at the top, the limit equals \(1/r\) by symmetry.)
	Finally, since strict majority implies plurality, \(P(\text{plurality for }1)\ge P(\text{majority for }1)\) for all \(n\).
	
\end{remark}

\section{CONCLUDING REMARKS}

We developed an analytical framework for collective accuracy in networked electorates with zealots, where free voters receive privately informative signals and update by local imitation. For finite populations, we obtain closed-form stationary distributions of the vote process and, as the number of free voters grows, characterize the corresponding vote-share limits. Under majority rule, both in binary and multi-alternative settings, we derive an exact large-population accuracy limit in closed form via the regularized incomplete beta function. Under plurality rule, we establish sharp closed-form lower bounds on accuracy, likewise expressed through regularized incomplete beta functions. An absolute-majority condition for the correct alternative guarantees that both majority and plurality accuracies strictly exceed the accuracy of any single voter, demonstrating amplification of informative signals through social interaction. These results extend the Condorcet jury theorem beyond independence and provide closed-form accuracy benchmarks for the design and analysis of social, biological, and engineered decision systems.

Further directions include moving beyond well-mixed interactions to structured topologies, including sparse, modular, and time-varying graphs, to clarify how connectivity modulates amplification and convergence; introducing agent-level heterogeneity, e.g., variation in signal quality, imitation rates, or local exposure to zealots, which may shift the relative performance of majority and plurality across parameter regimes; and developing finite-sample guarantees and concentration inequalities to guide decision thresholds, such as supermajorities, and the design of multi-stage aggregation schemes.

	\section{Acknowledgments}
	M.A.M.A. would like to thank FAPESP, grant 2021/14335-0, and CNPq, grant 303814/2023-3 for financial support.
	



\end{document}